\documentclass[12pt]{article}
\usepackage{amsmath,amssymb,amsthm}
\usepackage{graphicx}
\usepackage{empheq}
\usepackage{tikz}
\usetikzlibrary{arrows.meta}

\usepackage[hmargin=1in,vmargin=1in]{geometry}

\theoremstyle{plain}
\newtheorem{theorem}{Theorem}
\newtheorem{proposition}{Proposition}
\newtheorem{lemma}{Lemma}
\newtheorem{definition}{Definition}

\newtheoremstyle{rmk}{\topsep}{\topsep}{\upshape}{}{\bfseries}{.}{ }{}
\theoremstyle{rmk}
\newtheorem{remark}{Remark}

\newcommand{\xh}{\hat{x}}
\newcommand{\yh}{\hat{y}}
\newcommand{\xib}{\bar{\xi}}
\newcommand{\etab}{\bar{\eta}}

\title{\textbf{The Heavy Chain PDE:\\ Rapid Stabilization by Backstepping}}
\author{Miroslav Krstic\thanks{Department of Mechanical and Aerospace Engineering, University of California, San Diego, La Jolla, CA 92093-0411, USA. Email: \texttt{krstic@ucsd.edu}.}}
\date{}

\begin{document}
\maketitle

\begin{abstract}
We consider boundary stabilization of a heavy chain hanging from a moving trolley with no tip load. Because the tension vanishes at the free end, the wave speed vanishes there; in Riemann coordinates the model becomes a degenerate $2\times2$ hyperbolic system in which the coupling is singular and the free-end reflection is generated in the domain rather than by a boundary condition. We construct a Volterra backstepping transformation that maps this system to the same chain with uniform damping of an arbitrarily prescribed rate and an elastic restraint at the trolley, yielding exponential convergence of displacement, velocity, and strain to zero. The singular kernel equations are solved by selecting their bounded Frobenius branch at the free end, which replaces the missing boundary datum, and the four kernels are generated by a globally convergent power series. The transformation is boundedly invertible on the energy space, with inverse obtained by reversing the prescribed decay rate. The result extends the radial backstepping structure developed for parabolic equations on disks and balls to a degenerate hyperbolic system.
\end{abstract}

\newpage

\tableofcontents
\bigskip
\section{Introduction}

\subsection{Problem and literature}

A chain of length one hangs from a trolley and carries no load. Figure \ref{fig:chain} shows one swinging. Under the small-angle approximation, with gravity and length normalized, the horizontal deviation $u(x,t)$ obeys
\begin{equation}\label{eq:plant}
u_{tt}(x,t)=\bigl(x\,u_x(x,t)\bigr)_x, \qquad x\in(0,1),\ t>0,
\end{equation}
where $x$ is measured from the free end, and the actuation is the horizontal force applied at the trolley,
\begin{equation}\label{eq:input}
u_x(1,t)=F(t),
\end{equation}
the tension being normalized to one there. The collocated measurement is the velocity $u_t(1,t)$ of the actuated end, and every boundary law below is written as the force in terms of that velocity and of the distributed state.
At the free end the tension vanishes,
\begin{equation}\label{eq:tip}
\lim_{x\to0^+}x\,u_x(x,t)=0.
\end{equation}
\begin{figure}[t]
\centering
\includegraphics[width=0.5\textwidth]{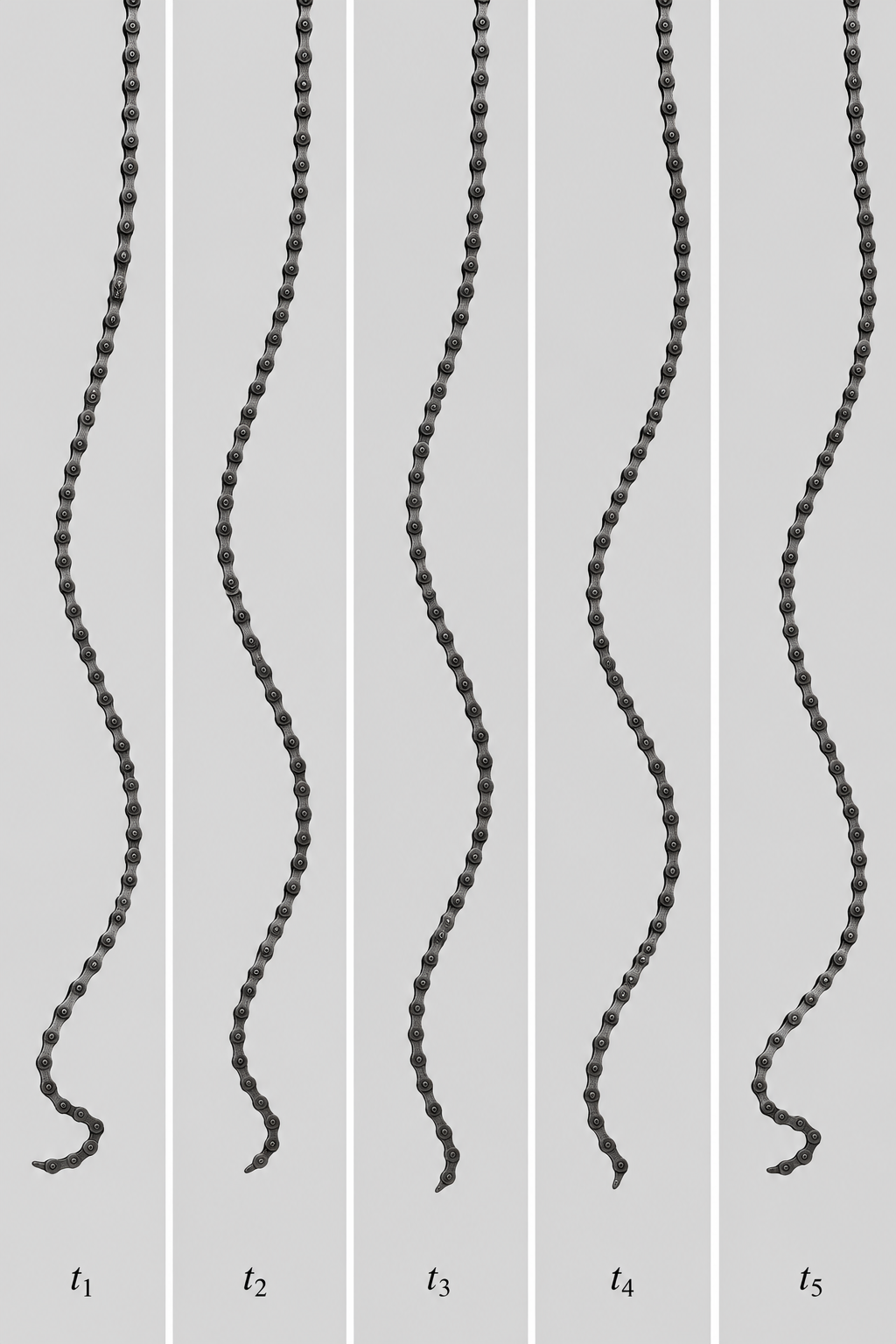}
\caption{Snapshots $t_1,\dots,t_5$ from an experiment with a bicycle chain hanging from a support. The oscillation takes the Bessel-function shapes of \eqref{eq:plant}: the wavelength is longest at the top, where the tension is largest, and shortens toward the tip, where the oscillations crowd together. The tension is the weight of the chain below and is therefore nonuniform along the chain, approaching zero at the free end, where the last links swing through the largest excursions.}
\label{fig:chain}
\end{figure}

The propagation speed is $\sqrt x$ and the transit time is $\int_0^1dx/\sqrt x=2$. In the Riemann invariants of Section~\ref{sec:riemann} the chain is a $2\times2$ hyperbolic system with unit speeds, and it is a degenerate member of that class: the speeds of \eqref{eq:plant} vanish at the uncontrolled end, and the degeneracy shows itself as a coupling between the two invariants that is singular there and a reflection carried by that coupling instead of by a boundary condition. The chain is the benchmark for this situation, and the difficulties met below belong to the class rather than to it.

The chain has been studied from the motion planning side. Murray \cite{Murray} treated it as a finite chain of pendulums and found it flat in the sense of \cite{FLMR}; Petit and Rouchon \cite{PetitRouchon} carried that to the infinite-dimensional model, parameterized the trajectories by those of the free end through distributed advances and delays, and showed that a transfer between rest states requires a time larger than twice the transit time.

A payload changes the problem. With a point mass at the free end the tension is bounded below by its weight, the model is uniformly hyperbolic, and boundary feedback obtained by backstepping through the dynamics of the cart gives four collocated gains under which the closed loop is exponentially stable, as St\"urzer, Arnold and Kugi \cite{Sturzer} establish. The load is what the normal-form designs surveyed in \cite{Gehring} take the hanging rope to have as well. The chain considered here carries none, and it is the vanishing of the tension at the free end that puts it outside that setting.

On the feedback side the work to measure against is that of Alabau-Boussouira, Cannarsa and Leugering \cite{ACL}, which established the control theory of the degenerate wave equations $u_{tt}=(x^{\mu}u_x)_x$, the chain being $\mu=1$. They established observability for $\mu<2$, and its failure beyond, exact controllability by the Hilbert uniqueness method, and exponential decay under boundary damping, for linear and for nonlinear feedbacks, with optimal decay rates; and they showed that the degeneracy does not degrade those rates at large time. Theirs is the feedback design for this plant, and the difficulty of the problem is measured by the fact that it stands alone: the line it opened has since gone to nondivergence form with drift \cite{BFM}, to degenerate beams, and alongside the degenerate parabolic literature \cite{CMV}, without a second construction for the chain itself. What their damper gives, and all that any single boundary damper can give, is decay at the rate the plant supplies.

Assigning the rate is a different matter. One route exists in principle: for a conservative system with an admissible boundary operator, exact controllability and a weighted Gramian yield a feedback of any prescribed rate \cite{Komornik,Urquiza,Vest}. It has not been carried out for the chain, which would require admissibility and exact controllability in the energy space to be established first, and it produces a feedback defined by inverting an operator assembled from trajectories, with no formula for its gains and no statement about what the controller does to the plant.

For the ordinary string the placement of that rate has an antecedent. With an anti-damper at the uncontrolled end all eigenvalues lie on a vertical line whose position is fixed by the reflection there, and Smyshlyaev and Krstic \cite{SmyshlyaevKrstic} move that line to any position the designer chooses, by a transformation acting on the velocity and the strain rather than on the displacement, and by a law that is a mismatched impedance at the actuated end together with the spatial average of the velocity. Their kernels are constants, the string being uniform. The chain is where that architecture meets a plant whose tension vanishes, and the kernels can no longer be constants.

Accounts of this kind, in which the closed loop is exhibited rather than estimated, are what backstepping supplies for $2\times2$ hyperbolic systems, where a Volterra transformation carries the plant to a target whose decay is evident \cite{KrsticSmyshlyaev,VazquezKrsticCoron,CoronVazquezKrsticBastin}, and the same architecture extends to one and to many counterconvecting states \cite{DiMeglioVazquezKrstic,HuDiMeglioVazquezKrstic}; \cite{Survey} surveys the field. The heavy chain is not in it, and the difference is not one of degree. In the Riemann variables that carry its energy, both invariants vanish at $x=0$; no boundary condition can be prescribed there; the reflection of a descending wave into an ascending one is performed by an in-domain coupling with a non-integrable singularity; and by Remark \ref{rem:invariant} no change of Riemann variables removes that singularity, so those designs do not apply after any change of coordinates. Singular coefficients of this kind have been met before in backstepping, and the body of work that met them is parabolic. Vazquez and Krstic began with the reaction-diffusion equation on a disk \cite{VK14}, extended it to balls of arbitrary dimension and to output feedback \cite{VK15,VK16balls}, treated the singular equation on the disk \cite{VazquezKrsticDisk} and the sphere under revolution symmetry \cite{VazquezKrsticSphere}, and with Zhang and Qi settled the well-posedness of the kernels and their computation by power series for radially dependent coefficients \cite{VZQK23}; the domain-extension and spatial-invariance methods that accompany this line are surveyed in \cite{Vazquez25}. In every one of these the radius makes the plant singular and the same singularity reappears in the kernel equations, and the decade of work required to dispose of it in the parabolic case is the measure of what the difficulty is worth. There the singularity is one of coordinates, the operator being uniformly elliptic and the radius singular only at the centre; the chain is degenerate as it stands, its speed vanishing at the free end. The two are nevertheless the same object: in the transit-time coordinate the chain is the radially symmetric wave equation on a disk of radius $2$, with the free end at its centre, as Remark \ref{rem:hardy} records. What that line met at the centre of a disk is what the degeneracy of the chain produces, and the hyperbolic counterpart of it is carried out here.

A degenerate plant has been stabilized rapidly once before, in the parabolic case: for $u_t=(x^\alpha u_x)_x$ with $\alpha\in(0,1)$, controlled at the end where the coefficient does not vanish, Gagnon, Lissy and Marx \cite{GagnonLissyMarx} assign any decay rate, with a target that is the same degenerate operator with $-\lambda$ added. The target here is of that kind. Their transformation is Fredholm, and establishing that it is continuous and invertible in the natural energy space is the work of that paper; the design below stays within the Volterra class, with an exponential multiplier.

For hyperbolic systems the transformation need not be of Volterra type: Coron, Hu and Olive \cite{CoronHuOlive} reach the optimal finite time for linear balance laws by a Volterra transformation followed by a Fredholm one, and the construction extends to coefficients depending on time and space \cite{CoronHuOliveShang}. Their in-domain coupling is bounded and their reflection is a boundary coefficient, neither of which the chain provides; the design given here stays inside the Volterra class, with an exponential multiplier.

\subsection{Idea of our design}\label{sec:idea}

The design sought here is stated in the physical variables. We seek a Volterra transformation of the state, together with the boundary law it induces, carrying \eqref{eq:plant}--\eqref{eq:tip} into the same chain with damping distributed along its length,
\begin{eqnarray}
\psi_{tt}+2\lambda\psi_t+\lambda^2\psi&=&\bigl(x\,\psi_x\bigr)_x,\qquad x\in(0,1),\label{eq:goal}\\
\lim_{x\to0^+}x\,\psi_x&=&0,\qquad \psi_x(1,t)\ =\ -c_0\,\psi(1,t),\label{eq:goalbc}
\end{eqnarray}
$\psi$ denoting the transformed displacement, $\lambda>0$ the decay rate and $c_0>0$ the stiffness, both at the designer's choice. Every material point of \eqref{eq:goal} acquires viscous damping $2\lambda$ and the stiffness correction $\lambda^2$; the tension condition at the free end is untouched, and the support is elastically restrained rather than held. The restraint is what leaves no equilibrium but rest: the static solutions of \eqref{eq:goal} are $c\,I_0(2\lambda\sqrt x)$, and \eqref{eq:goalbc} forces $c\bigl[\lambda I_1(2\lambda)+c_0I_0(2\lambda)\bigr]=0$, hence $c=0$. Without it the chain would straighten but keep a rigid horizontal offset, which carries no energy and which no feedback in $u_t$ and $u_x$ alone can remove. The chain is still a chain, hanging under its own weight with a free end: the physics is changed and not blurred.

On this chain the delay from a material point to the trolley is $2(1-\sqrt x)$, not proportional to the distance: the delay per unit length grows without bound toward the free end, so equal lengths of chain occupy unequal shares of the actuator's future. Asking for one rate at every height against such a schedule leaves no freedom in the gains: backstepping returns them graded along the length, the grading dictated by the plant.

\subsection{Contributions}

\begin{enumerate}
\item \emph{Backstepping for plants made singular by a radial coordinate is advanced from parabolic to hyperbolic PDEs.} On disks and balls the method has been established for reaction-diffusion equations \cite{VK14,VK16balls,VazquezKrsticDisk,VazquezKrsticSphere,VZQK23}. For wave equations the same singularity has been an obstacle. The heavy chain is its simplest instance: in the transit-time coordinate the chain is the wave equation on a disk, with the free end at the centre.

\item \emph{Backstepping is extended to a $2\times2$ hyperbolic system that is degenerate.} Existing designs require bounded coupling coefficients and a reflection coefficient at the uncontrolled boundary \cite{VazquezKrsticCoron,CoronVazquezKrsticBastin,DiMeglioVazquezKrstic,HuDiMeglioVazquezKrstic}. The chain has neither: both speeds vanish at the free end, the coupling is unbounded there, and the reflection is produced by that coupling. No change of Riemann variables restores either property. The design that succeeds does not remove the singular term; it retains the conversion of one wave into the other at the free end, adds the damping around it, and restrains the support, so that the two transports and the boundary state are strictly passive in the same pair of ports, for every rate and every stiffness.

\item \emph{The kernel equations, singular at the free end and short of a boundary condition, are solved.} There the four equations admit no boundary datum. Their solutions are of two types, one growing as the inverse square root of the distance to that end, one vanishing as its square root, and boundedness selects the second: that requirement replaces the missing condition. The kernels then follow from a series in the decay rate, convergent for every rate, with polynomial coefficients. Series of this type resolved the radial parabolic kernels \cite{VZQK23}; the problem they resolve here is not that one.

\item \emph{The decay rate of the chain is specified by the designer, and the closed loop is a chain with that rate.} A boundary damper dissipates what reaches it, at a rate the plant determines \cite{ACL}. We construct instead a full-state feedback whose gain profiles are the kernels of an invertible Volterra transformation, and with it a similarity between the closed loop and a chain that hangs under its own weight, reflects at its free end as before, and is damped along its length at any rate specified. The inverse transformation is the same one with the rate reversed. Relocating the decay rate of a wave equation by a transformation of this kind was done for the uniform string in \cite{SmyshlyaevKrstic}, with constant kernels; the chain admits no such collapse.
\end{enumerate}

\section{Energy and Riemann variables}\label{sec:energy}

\begin{table}[t]
\centering
\renewcommand{\arraystretch}{1.45}
\begin{tabular}{@{}l l l@{}}
\hline
\noalign{\vskip 3pt}
Pair & Definition & Representation\\
\noalign{\vskip 3pt}
\hline
\noalign{\vskip 5pt}
$\bigl(\sqrt x\,u_x,\ u_t\bigr)$ & \eqref{eq:state} & plant state\\
$(\tau,v)$ & $\tau=x\,u_x$, $v=u_t$, \eqref{eq:physical} & first-order physical form\\
$(y,v)$ & $y=\tau/\sqrt x$, \eqref{eq:zyscaled} & force scaled, equal speeds\\
$(\xi,\eta)$ & $(y\pm v)/2$, \eqref{eq:sysxieta} & invariants, diagonal convection\\
$(\xib,\etab)$ & $\sqrt{\xh}$ times $(\xi,\eta)$, \eqref{eq:riemanndef} & plant in Riemann form \eqref{eq:sys}\\
$(\alpha,\beta,\sigma)$ & $\mathcal T_\lambda(\xib,\etab)$, $(\alpha+\beta)(2,t)$ & target \eqref{eq:tgt}, \eqref{eq:tgtbcode}\\
$(\xib_d,\etab_d)$ & \eqref{eq:riemanndisp} & Riemann form of $(w,u)$\\
$(\alpha_d,\beta_d)$ & $\mathcal T_\lambda(\xib_d,\etab_d)$, \eqref{eq:transfdisp} & antiderivatives of $(\alpha,\beta)$ in time\\
$(w,u)$ & $w=p/\sqrt x$, \eqref{eq:impulse}--\eqref{eq:wdef} & displacement realization of the plant\\
$(\omega,\psi)$ & \eqref{eq:psiomegadef} & displacement realization of the target\\
\noalign{\vskip 5pt}
\hline
\end{tabular}
\caption{Representations of the chain. The rows above the last two carry velocity and strain, in the plant and in the target; the last two carry displacement, and $(w,u)$ stands to $(\omega,\psi)$ as $(\xib,\etab)$ stands to $(\alpha,\beta)$. The target is $(\alpha,\beta,\sigma)$ together with $(\alpha_d,\beta_d)$ and $(\omega,\psi)$; everything else is the plant.}
\label{tab:vars}
\end{table}

\label{sec:riemann}

The state of \eqref{eq:plant} is the pair of scaled strain and velocity,
\begin{equation}\label{eq:state}
Z(t)=\Bigl(\sqrt x\,u_x(\cdot,t),\ u_t(\cdot,t)\Bigr),
\end{equation}
and the mechanical energy is
\begin{equation}\label{eq:energy}
E(t)=\frac12\int_0^1\Bigl(x\,u_x^2(x,t)+u_t^2(x,t)\Bigr)dx .
\end{equation}
On \eqref{eq:state} alone the displacement is invisible: a rigid translation $u\equiv c$ has neither velocity nor strain, is an equilibrium of \eqref{eq:plant} with $F=0$, and gives $E=0$. Since the aim \eqref{eq:goal}--\eqref{eq:goalbc} is to drive $u$ itself to rest, the norm carries the trolley position as well,
\begin{equation}\label{eq:normc0}
\bigl\|(u,u_t)\bigr\|^2=\int_0^1\Bigl(x\,u_x^2+u_t^2\Bigr)dx+c_0\,u^2(1,t),
\end{equation}
which is a norm and not a seminorm: $\int xu_x^2=0$ with $u(1)=0$ forces $u\equiv0$.

Throughout, $x$ is traded for the transit-time coordinate
\begin{equation}\label{eq:xhdef}
\xh=2\sqrt x\in(0,2),
\end{equation}
in which the chain travels at unit speed. In it,

\begin{equation}\label{eq:supnorm}
\|Z\|_{\mathcal X}=\sup_{\xh\in(0,2)}\Bigl(|u_{\xh}(\xh)|+|u_t(\xh)|\Bigr)
\end{equation}
is the supremum norm of the gradient of a radially symmetric field on a disk of radius $2$, the chain in that coordinate being the axisymmetric wave equation there. Since the domain is bounded, $2E\le\|Z\|_{\mathcal X}^2$.

\begin{definition}\label{def:class}
$\mathcal X_0$ is the set of states \eqref{eq:state} with $\|Z\|_{\mathcal X}<\infty$, $u_t$ and $u_{\xh}$ continuous on $[0,2]$, and
\begin{equation}\label{eq:class}
u_{\xh}(0)=0,\qquad\text{equivalently}\qquad \lim_{x\to0^+}\sqrt x\,u_x(x)=0 .
\end{equation}
\end{definition}

\begin{remark}\label{rem:focusing}
A state with $u_{\xh}(0)\ne0$ is a cone at the centre of the disk. On such states the norm \eqref{eq:supnorm} is amplified by focusing, and no estimate $\|Z(t)\|_{\mathcal X}\le c\|Z(0)\|_{\mathcal X}$ holds, for the open loop or for any feedback.
\end{remark}

\begin{lemma}[Riemann variables and energy]\label{lem:riemann}
Under $\xh=2\sqrt x$ and
\begin{empheq}[box=\fbox]{equation}\label{eq:riemanndef}
\xib=\frac{x^{1/4}}{\sqrt2}\Bigl(\sqrt x\,u_x+u_t\Bigr),\qquad
\etab=\frac{x^{1/4}}{\sqrt2}\Bigl(\sqrt x\,u_x-u_t\Bigr),
\end{empheq}
the heavy chain \eqref{eq:plant}--\eqref{eq:tip} reads
\begin{subequations}\label{eq:sys}
\begin{align}
\xib_t&=\xib_{\xh}+\frac{1}{2\xh}\,\etab,\label{eq:sysxi}\\
\etab_t&=-\etab_{\xh}-\frac{1}{2\xh}\,\xib,\label{eq:syseta}
\end{align}
\end{subequations}
the map \eqref{eq:riemanndef} is an isometry,
\begin{equation}\label{eq:isometry}
\int_0^2\bigl(\xib^2+\etab^2\bigr)d\xh=\int_0^1\bigl(x\,u_x^2+u_t^2\bigr)dx=2E,
\end{equation}
and every state in $\mathcal X_0$ satisfies
\begin{equation}\label{eq:tipbehavior}
\xib(\xh)=\frac{\sqrt{\xh}}{2}\,u_t(0)+o\bigl(\sqrt{\xh}\bigr),
\qquad
\etab(\xh)=-\frac{\sqrt{\xh}}{2}\,u_t(0)+o\bigl(\sqrt{\xh}\bigr).
\end{equation}
In particular both invariants vanish at $\xh=0$, where no boundary condition is imposed, and the energy flux $\xib^2-\etab^2$ vanishes there.
\end{lemma}
\emph{Proof.} With the velocity $v=u_t$ and the tension force $\tau=x\,u_x$, the chain \eqref{eq:plant} is the pair
\begin{subequations}\label{eq:physical}
\begin{align}
\tau_t&=x\,v_x,\label{eq:physical1}\\
v_t&=\tau_x,\label{eq:physical2}
\end{align}
\end{subequations}
the first-order form in which the spatial derivatives sit off the diagonal: linear acoustics, with $\tau$ in the role of the pressure and $x$ in that of the modulus, vanishing at the free end. Scaling the force, $y=\tau/\sqrt x=\sqrt x\,u_x$, puts the same factor $\sqrt x$ on both spatial derivatives,
\begin{subequations}\label{eq:zyscaled}
\begin{align}
y_t&=\sqrt x\,v_x,\label{eq:zyscaled1}\\
v_t&=\sqrt x\,y_x+\frac{y}{2\sqrt x},\label{eq:zyscaled2}
\end{align}
\end{subequations}
at the price of the one term that breaks the symmetry between the two equations. The symmetric and antisymmetric combinations $\xi=(v+y)/2$ and $\eta=(y-v)/2$ restore it, carrying the convection onto the diagonal and the singular term off it,
\begin{subequations}\label{eq:xietax}
\begin{align}
\xi_t&=\ \ \sqrt x\,\xi_x+\frac{\xi+\eta}{4\sqrt x},\label{eq:xietax1}\\
\eta_t&=-\sqrt x\,\eta_x-\frac{\xi+\eta}{4\sqrt x}.\label{eq:xietax2}
\end{align}
\end{subequations}
The substitution $\xh=2\sqrt x$ turns $\sqrt x\,\partial_x$ into $\partial_{\xh}$ and $1/(4\sqrt x)$ into $1/(2\xh)$, giving unit speeds,
\begin{subequations}\label{eq:sysxieta}
\begin{align}
\xi_t&=\ \ \xi_{\xh}+\frac{\xi+\eta}{2\xh},\label{eq:sysxieta1}\\
\eta_t&=-\eta_{\xh}-\frac{\xi+\eta}{2\xh},\label{eq:sysxieta2}
\end{align}
\end{subequations}
and $\xib=\sqrt{\xh}\,\xi$, $\etab=\sqrt{\xh}\,\eta$, for which $\sqrt{\xh}\,\xi_{\xh}=\xib_{\xh}-\xib/(2\xh)$, remove the diagonal terms and give \eqref{eq:sys}; writing $\xh=2\sqrt x$ in $\xib=\sqrt{\xh}(y+v)/2$ gives \eqref{eq:riemanndef}. For \eqref{eq:isometry}, $\xib^2+\etab^2=\xh(y^2+v^2)/2=\sqrt x\,(x\,u_x^2+u_t^2)$ and $d\xh=dx/\sqrt x$. For \eqref{eq:tipbehavior}, $u_{\xh}(0)=0$ by \eqref{eq:class} and $\xib=\sqrt{\xh}(u_t+u_{\xh})/2$. \hfill$\square$

The displacement enters the design through a second realization of the same dynamics. Let $p$ be the impulse potential of \eqref{eq:physical},
\begin{equation}\label{eq:impulse}
p_x=u_t,\qquad p_t=x\,u_x,\qquad p(0,t)=0,
\end{equation}
which exists because $\partial_t(x u_x)=\partial_x u_t$ is \eqref{eq:plant}, and let
\begin{equation}\label{eq:wdef}
w(x,t)=\frac{p(x,t)}{\sqrt x}=\frac{1}{\sqrt x}\int_0^x u_t(s,t)\,ds .
\end{equation}
Then
\begin{equation}\label{eq:uwsystem}
w_t=\sqrt x\,u_x,\qquad u_t=\sqrt x\,w_x+\frac{w}{2\sqrt x},
\end{equation}
the second because the two terms in $w/(2\sqrt x)$ cancel; comparing with \eqref{eq:zyscaled}, the pair $(w,u)$ satisfies the same system as the pair $(\sqrt x\,u_x,u_t)$, with the displacement in the slot the velocity occupies there. One transformation therefore serves both, and the displacement is carried by the second.
In that coordinate $\sqrt x\,\partial_x=\partial_{\xh}$, so the second component of \eqref{eq:state} is $u_{\xh}$, and

\begin{remark}\label{rem:hardy}
In the coordinate $\xh$ the chain is the radially symmetric two-dimensional wave equation on a disk of radius $2$, and the free end is its centre.
\end{remark}

\begin{remark}\label{rem:elsewhere}
The degeneracy is not peculiar to the chain, and neither is \eqref{eq:physical}. In a fluid column at rest under gravity, with $x$ the distance below the free surface, the hydrostatic pressure is the weight of the column above and vanishes linearly at the surface; the square of the sound speed vanishes with it, and linearized acoustics at uniform background density is \eqref{eq:physical}, with $\tau$ the pressure perturbation and $v$ the velocity. The chain's tension is the weight below and vanishes at the free end in the same way. The operator \eqref{eq:plant} arises as well at a shoreline of constant slope, where the shallow water equations degenerate with the depth in place of the tension \cite{CarrierGreenspan}, and, by Remark \ref{rem:hardy}, on a disk.
\end{remark}

\begin{remark}\label{rem:invariant}
No change of Riemann variables bounds the coupling. Rescaling by $\tilde\xi=\varphi\,\xib/\sqrt{\xh}$ and $\tilde\eta=\psi\,\etab/\sqrt{\xh}$, with $\varphi,\psi$ positive and $C^1$, leaves \eqref{eq:sys} in the same form with couplings $c_1=\varphi/(2\psi\xh)$ and $c_2=-\psi/(2\varphi\xh)$, whose product
\begin{equation}\label{eq:product}
c_1(\xh)\,c_2(\xh)=-\frac{1}{4\xh^2}
\end{equation}
is the same for every $\varphi,\psi$. Hence $\max\{|c_1|,|c_2|\}\ge1/(2\xh)$ always, the chain does not belong to the bounded-coefficient $2\times2$ class in any Riemann coordinates, and the designs of \cite{VazquezKrsticCoron,CoronVazquezKrsticBastin} do not apply to it after any change of variables. Diagonal rescalings are the only ones available, since a pointwise linear map leaving the principal part equal to $\operatorname{diag}(1,-1)$ preserves its eigenspaces, the speeds being distinct.
\end{remark}

\section{Controller and closed-loop stability with assignable decay}\label{sec:rapid}

\subsection{Target and controller}\label{sec:target}

\begin{figure}[t]
\centering
\begin{tikzpicture}[>=stealth,line width=0.5pt,font=\small]
  \tikzstyle{blk}=[draw,rounded corners=3pt,minimum height=11mm]
  \node[blk,minimum width=68mm] (A) at (0,2.6) {$\alpha$: \textbf{leftward} transport with singular skew-coupling};
  \node[blk,minimum width=68mm] (B) at (0,0)   {$\beta$: \textbf{rightward} transport with singular skew-coupling};
  \foreach \i in {0,...,11}{
    \pgfmathsetmacro{\xx}{-3.1+0.56*\i}
    \ifodd\i
      \draw[->] (\xx,2.05) -- (\xx,0.56);
    \else
      \draw[->] (\xx,0.56) -- (\xx,2.05);
    \fi}
  \node[draw,circle,inner sep=1.6pt] (S) at (5.8,2.6) {$+$};
  \node[blk,minimum width=24mm,minimum height=18mm,align=center] (K) at (8.0,2.6)
    {$\sigma$ block\\[3pt]$\dfrac{2c_0}{s+\lambda+c_0}$};
  \draw[->] (B.east) -- (8.0,0) node[midway,below]{\footnotesize$\beta(2,t)$} -- (K.south);
  \fill (5.8,0) circle (1.4pt);
  \draw[->] (5.8,0) -- node[left,pos=0.75]{$-$} (S);
  \draw[->] (K.west) -- node[above]{\footnotesize$\sigma$} (S);
  \draw[->] (S) -- node[above,pos=0.5]{\footnotesize$\alpha(2,t)$} (A.east);
\end{tikzpicture}
\caption{The target \eqref{eq:tgt}, \eqref{eq:tgtbcode} as a block diagram. The two transports counterconvect and are coupled along the whole domain by the singular term, which is the free-end reflection; at the trolley the outgoing $\beta$ drives one first-order boundary state and is differenced with it to form the incoming $\alpha$. The block from $(\alpha-\beta)(2,t)$ to $\sigma$, going through the $\sigma$ block, is strictly passive, and so is the block from $\sigma$ to $-(\alpha-\beta)(2,t)$, going through the two transports.}
\label{fig:blocks}
\end{figure}
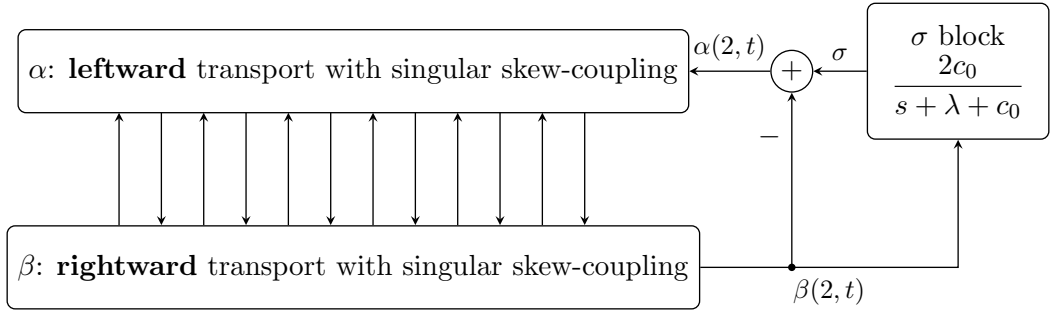

The plant \eqref{eq:sys} is two counterpropagating waves converted into one another, at the free end, by the singular coupling $1/(2\xh)$. The target keeps both,
\begin{subequations}\label{eq:tgt}
\begin{align}
\alpha_t&=\ \ \alpha_{\xh}+\frac{1}{2\xh}\,\beta-\lambda\alpha,\label{eq:tgta}\\
\beta_t&=-\beta_{\xh}-\frac{1}{2\xh}\,\alpha-\lambda\beta,\label{eq:tgtb}
\end{align}
\end{subequations}
with the spring \eqref{eq:goalbc} at the trolley, which in these variables relates the two traces, as in Figure \ref{fig:blocks},
\begin{subequations}\label{eq:tgtbcode}
\begin{align}
\dot\sigma+(\lambda+c_0)\,\sigma&=2c_0\,\beta(2,t),\label{eq:tgtbcsigma}\\
\alpha(2,t)&=\sigma(t)-\beta(2,t).\label{eq:tgtbcalpha}
\end{align}
\end{subequations}
It is the chain's own first-order model, with the free-end conversion mechanism unchanged, a uniform damping $-\lambda$ added, and the support elastically restrained. The coupling is conservative, the flux $\alpha^2-\beta^2$ vanishes at the free end, and at the trolley $\alpha^2-\beta^2=\sigma(\alpha-\beta)(2,t)=-\bigl(\sigma\dot\sigma+\lambda\sigma^2\bigr)/c_0$ by \eqref{eq:tgtbcode}, so that the Lyapunov derivative satisfies not an inequality but the identity
\begin{equation}\label{eq:tgtdecay}
\frac{d}{dt}\Bigl[\int_0^2\bigl(\alpha^2+\beta^2\bigr)d\xh+\frac{\sigma^2}{2c_0}\Bigr]
=-2\lambda\Bigl[\int_0^2\bigl(\alpha^2+\beta^2\bigr)d\xh+\frac{\sigma^2}{2c_0}\Bigr].
\end{equation}

This is \eqref{eq:goal}--\eqref{eq:goalbc} in the Riemann variables.

The transformation is an exponential multiplier and a Volterra correction,
\begin{subequations}\label{eq:transf}
\begin{align}
\alpha(\xh,t)&=e^{\lambda\xh}\Bigl[\xib(\xh,t)-\int_0^{\xh}\Bigl(K^{\xi\xi}\xib(\yh,t)+K^{\xi\eta}\etab(\yh,t)\Bigr)d\yh\Bigr],\label{eq:transfa}\\
\beta(\xh,t)&=e^{-\lambda\xh}\Bigl[\etab(\xh,t)-\int_0^{\xh}\Bigl(K^{\eta\xi}\xib(\yh,t)+K^{\eta\eta}\etab(\yh,t)\Bigr)d\yh\Bigr],\label{eq:transfb}
\end{align}
\end{subequations}
the kernels being evaluated at $(\xh,\yh)$. The arguments $\xib,\etab$ are the Riemann variables \eqref{eq:riemanndef}, which in the coordinate \eqref{eq:xhdef} read
\begin{equation}\label{eq:riemannxh}
\xib=\frac{\sqrt{\xh}}{2}\Bigl(u_{\xh}+u_t\Bigr),\qquad
\etab=\frac{\sqrt{\xh}}{2}\Bigl(u_{\xh}-u_t\Bigr).
\end{equation}
The two steps of \eqref{eq:transf} do separate jobs. The multiplier creates the damping $-\lambda$ and, in creating it, changes the two coupling coefficients from $\pm1/(2\xh)$ to $e^{-2\lambda\xh}/(2\xh)$ and $-e^{2\lambda\xh}/(2\xh)$; the kernels install that tilt. The product of the two couplings is
\begin{equation}\label{eq:tilt}
\frac{e^{-2\lambda\xh}}{2\xh}\cdot\Bigl(-\frac{e^{2\lambda\xh}}{2\xh}\Bigr)=-\frac{1}{4\xh^2},
\end{equation}
the plant's own, so the strength of the free-end conversion is untouched and only its apportionment between the two directions is changed. This is what a design that removed the coupling cannot do.

By \eqref{eq:uwsystem} the pair $(w,u)$ obeys the same system as the state, so \eqref{eq:transf} applies to it with the same four kernels. Carrying it out in full: form
\begin{empheq}[box=\fbox]{equation}\label{eq:riemanndisp}
\xib_d=\frac{x^{1/4}}{\sqrt2}\Bigl(u+\frac{1}{\sqrt x}\int_0^xu_t\,ds\Bigr),\qquad
\etab_d=\frac{x^{1/4}}{\sqrt2}\Bigl(-u+\frac{1}{\sqrt x}\int_0^xu_t\,ds\Bigr),
\end{empheq}
which is \eqref{eq:riemanndef} with $(w,u)$ in place of $(\sqrt x\,u_x,u_t)$, transform them by \eqref{eq:transf},
\begin{subequations}\label{eq:transfdisp}
\begin{align}
\alpha_d(\xh,t)&=e^{\lambda\xh}\Bigl[\xib_d(\xh,t)-\int_0^{\xh}\Bigl(K^{\xi\xi}\xib_d(\yh,t)+K^{\xi\eta}\etab_d(\yh,t)\Bigr)d\yh\Bigr],\label{eq:transfdispa}\\
\beta_d(\xh,t)&=e^{-\lambda\xh}\Bigl[\etab_d(\xh,t)-\int_0^{\xh}\Bigl(K^{\eta\xi}\xib_d(\yh,t)+K^{\eta\eta}\etab_d(\yh,t)\Bigr)d\yh\Bigr],\label{eq:transfdispb}
\end{align}
\end{subequations}
which by the intertwining satisfy \eqref{eq:tgta}--\eqref{eq:tgtb}, the same target equations as $\alpha,\beta$, of which they are the antiderivatives in time; and undo \eqref{eq:riemanndisp} on the result:
\begin{equation}\label{eq:psiomegadef}
\psi(x,t)=\frac{\alpha_d-\beta_d}{\sqrt{\xh}},\qquad
\omega(x,t)=\frac{\alpha_d+\beta_d}{\sqrt{\xh}},\qquad \xh=2\sqrt x .
\end{equation}
So $\psi$ and $\omega$ are functionals of $u(\cdot,t)$ and $u_t(\cdot,t)$ at the present instant, and $\psi$ carries $u$ itself: replacing $u$ by $u+c$ changes $\psi$ and with it the control, which is how the design sees the rigid translation. They satisfy
\begin{equation}\label{eq:psiomega}
\omega_t=\sqrt x\,\psi_x-\lambda\omega,\qquad
\psi_t=\sqrt x\,\omega_x+\frac{\omega}{2\sqrt x}-\lambda\psi,
\end{equation}
and eliminating $\omega$ returns the target \eqref{eq:goal} in the physical variables. Since $\sqrt x\,\psi_x=\omega_t+\lambda\omega$, the spring \eqref{eq:goalbc} at the trolley is
\begin{equation}\label{eq:tgtbc}
\omega_t(1,t)+\lambda\,\omega(1,t)+c_0\,\psi(1,t)=0,
\end{equation}
of which \eqref{eq:tgtbcode} is the time derivative, the state pair being the time derivative of $(w,u)$.

The kernels are given by a series. With $\sigma=\yh/\xh$,
\begin{equation}\label{eq:Kseries}
K^{\xi\xi}(\xh,\yh)=\sqrt{\frac{\yh}{\xh}}\sum_{n\ge1}\lambda^n\xh^{\,n-1}a_n(\sigma),
\end{equation}
and the same with $b_n,c_n,d_n$ in place of $a_n$ for $K^{\xi\eta},K^{\eta\xi},K^{\eta\eta}$, where each coefficient is a polynomial of degree at most $n-1$ determined recursively by \eqref{eq:Rn}--\eqref{eq:Rndiag}, the first two orders being
\begin{equation}\label{eq:Klead}
a_1=c_1=\tfrac12,\quad b_1=d_1=-\tfrac12,\qquad
a_2=d_2=\tfrac{\sigma-1}{4},\quad b_2=c_2=\tfrac{\sigma+1}{4}.
\end{equation}
The factor $\sqrt{\yh/\xh}$ is not an estimate appended to the solution: it is the regular Frobenius branch that the degeneracy at the free end selects, and what remains after it is removed is polynomial at every order. The series converges for every $\lambda$ on the whole triangle $0<\yh<\xh<2$, so the kernels are constructive: no two-dimensional problem has to be solved numerically.

\begin{table}[t]
\centering
\fbox{\begin{minipage}{0.92\textwidth}
\medskip
\noindent\textbf{Modules in feedback \eqref{eq:rapidfb}:}
\begin{subequations}\label{eq:PsiOmega}
\begin{align}
\Psi&=\cosh(2\lambda)\,u(1,t)+\sinh(2\lambda)\!\int_0^1\!u_t\,dx
+\frac{e^{-2\lambda}J_2-e^{2\lambda}J_1}{\sqrt2},\label{eq:Psidef}\\
\Omega&=\sinh(2\lambda)\,u(1,t)+\cosh(2\lambda)\!\int_0^1\!u_t\,dx
-\frac{e^{2\lambda}J_1+e^{-2\lambda}J_2}{\sqrt2},\label{eq:Omegadef}
\end{align}
\end{subequations}
\begin{subequations}\label{eq:Idef}
\begin{align}
\binom{I_1}{I_2}&=\int_0^2\begin{pmatrix}K^{\xi\xi}&K^{\xi\eta}\\ K^{\eta\xi}&K^{\eta\eta}\end{pmatrix}(2,\yh)\binom{\xib}{\etab}\,d\yh,\label{eq:Idefa}\\[2pt]
\binom{J_1}{J_2}&=\int_0^2\begin{pmatrix}K^{\xi\xi}&K^{\xi\eta}\\ K^{\eta\xi}&K^{\eta\eta}\end{pmatrix}(2,\yh)\binom{\xib_d}{\etab_d}\,d\yh,\label{eq:Idefb}
\end{align}
\end{subequations}
\end{minipage}}

\medskip

\fbox{\begin{minipage}{0.92\textwidth}
\medskip
\noindent\textbf{Kernels in \eqref{eq:PsiOmega}, \eqref{eq:Idef}, \eqref{eq:rapidfb}:}
\begin{subequations}\label{eq:Ksystem}
\begin{eqnarray}
K^{\xi\xi}_{\xh}+K^{\xi\xi}_{\yh}&=&-\frac{1}{2\yh}K^{\xi\eta}-\frac{e^{-2\lambda\xh}}{2\xh}K^{\eta\xi},\label{eq:Kxx}\\
K^{\xi\eta}_{\xh}-K^{\xi\eta}_{\yh}&=&\ \ \frac{1}{2\yh}K^{\xi\xi}-\frac{e^{-2\lambda\xh}}{2\xh}K^{\eta\eta},\label{eq:Kxy}\\
K^{\eta\xi}_{\xh}-K^{\eta\xi}_{\yh}&=&\ \ \frac{1}{2\yh}K^{\eta\eta}-\frac{e^{2\lambda\xh}}{2\xh}K^{\xi\xi},\label{eq:Kyx}\\
K^{\eta\eta}_{\xh}+K^{\eta\eta}_{\yh}&=&-\frac{1}{2\yh}K^{\eta\xi}-\frac{e^{2\lambda\xh}}{2\xh}K^{\xi\eta},\label{eq:Kyy}\\
K^{\xi\eta}(\xh,\xh)&=&\frac{e^{-2\lambda\xh}-1}{4\xh},\qquad
K^{\eta\xi}(\xh,\xh)\ =\ \frac{e^{2\lambda\xh}-1}{4\xh},\label{eq:Kdiag}
\end{eqnarray}
\end{subequations}
\end{minipage}}
\caption{The feedback law \eqref{eq:rapidfb} in full. Above: the two trolley functionals and the four kernel integrals, whose arguments are the Riemann variables \eqref{eq:riemanndef} and \eqref{eq:riemanndisp} in the coordinate \eqref{eq:xhdef}. Below: the kernel equations on $0<\yh<\xh<2$, with the diagonal data.}
\label{tab:law}
\end{table}

\begin{theorem}[Rapid stabilization]\label{thm:rapid}
Let $\lambda>0$ and $c_0>0$, let the kernels of \eqref{eq:transf} solve \eqref{eq:Ksystem} of Table \ref{tab:law} and be bounded at $\yh=0$ with $|K|\le C_\lambda\sqrt{\yh/\xh}$, and let the force at the trolley be
\begin{equation}\label{eq:rapidfb}
u_x(1,t)=\frac{1}{\cosh 2\lambda}\Bigl[-c_0\Psi(t)-\lambda\,\Omega(t)
-u_t(1,t)\sinh2\lambda+\frac{e^{2\lambda}I_1(t)+e^{-2\lambda}I_2(t)}{\sqrt2}\Bigr],
\end{equation}
with $\Psi,\Omega$ and $I_1,I_2,J_1,J_2$ given by \eqref{eq:PsiOmega}--\eqref{eq:Idef} of Table \ref{tab:law}. Every quantity there, along with \eqref{eq:xhdef}, \eqref{eq:riemanndef} and \eqref{eq:riemanndisp}, is an integral over the chain of the profiles $u(\cdot,t)$ and $u_t(\cdot,t)$ at the present instant. Then \eqref{eq:transf} carries the plant into \eqref{eq:tgt} with \eqref{eq:tgtbc} and is boundedly invertible; for every initial state the closed loop has a unique mild solution, classical for data in the domain of its generator; and
\begin{equation}\label{eq:ratedecay}
\bigl\|(u,u_t)(t)\bigr\|\le M_{\lambda,c_0}\,e^{-\lambda t}\,\bigl\|(u,u_t)(0)\bigr\|,\qquad t\ge0,
\end{equation}
in the norm \eqref{eq:normc0}, for some $M_{\lambda,c_0}\ge1$.
\end{theorem}

\emph{Proof.} Write $a$ and $b$ for the bracketed quantities in \eqref{eq:transf}, so that $\alpha=e^{\lambda\xh}a$ and $\beta=e^{-\lambda\xh}b$, and put $c(\yh)=\frac1{2\yh}$, $p(\xh)=\frac{e^{-2\lambda\xh}}{2\xh}$, $q(\xh)=\frac{e^{2\lambda\xh}}{2\xh}$. Substituting into \eqref{eq:tgt}, the multiplier removes the damping and tilts the couplings, so that \eqref{eq:tgt} holds if and only if
\begin{equation}\label{eq:intermediate}
a_t=a_{\xh}+p(\xh)\,b,\qquad b_t=-b_{\xh}-q(\xh)\,a .
\end{equation}
Differentiating $a$ in $t$, substituting \eqref{eq:sys} and integrating by parts,
\begin{eqnarray}
a_t-a_{\xh}-p\,b&=&\Bigl(c(\xh)+2K^{\xi\eta}(\xh,\xh)-p(\xh)\Bigr)\etab(\xh)
+\lim_{\yh\to0^+}\Bigl(K^{\xi\xi}\xib-K^{\xi\eta}\etab\Bigr)\nonumber\\
&&+\int_0^{\xh}\Bigl(K^{\xi\xi}_{\xh}+K^{\xi\xi}_{\yh}+c(\yh)K^{\xi\eta}+p(\xh)K^{\eta\xi}\Bigr)\xib\,d\yh\nonumber\\
&&+\int_0^{\xh}\Bigl(K^{\xi\eta}_{\xh}-K^{\xi\eta}_{\yh}-c(\yh)K^{\xi\xi}+p(\xh)K^{\eta\eta}\Bigr)\etab\,d\yh ,\label{eq:identA}
\end{eqnarray}
the traces of $K^{\xi\xi}$ on the diagonal cancelling between the two derivatives. The boundary term vanishes, since $K^{\xi\xi},K^{\xi\eta}=O(\sqrt{\yh})$ and $\xib,\etab=O(\sqrt{\yh})$ by \eqref{eq:tipbehavior}, and annihilating the other three groups, together with the same computation on $b_t+b_{\xh}+q\,a$, gives \eqref{eq:Ksystem}--\eqref{eq:Kdiag}, whose solution, with the kernels bounded at $\yh=0$, is \eqref{eq:Kseries}; existence, the bound and the invertibility are Proposition \ref{prop:existence} and Lemma \ref{lem:HS}. For the decay, multiply \eqref{eq:tgta} by $\alpha$ and \eqref{eq:tgtb} by $\beta$ and add: the coupling terms cancel, the flux integrates to $\bigl[\alpha^2-\beta^2\bigr]_0^2$, which vanishes at $\xh=0$ since $\alpha,\beta=O(\sqrt{\xh})$ and equals $-\bigl(\sigma\dot\sigma+\lambda\sigma^2\bigr)/c_0$ at $\xh=2$ by \eqref{eq:tgtbcode}, and \eqref{eq:tgtdecay} follows. For the control law, \eqref{eq:transf} at the trolley gives
\begin{equation}\label{eq:omegat1}
\omega_t(1,t)=\cosh(2\lambda)\,u_x(1,t)+\sinh(2\lambda)\,u_t(1,t)-\frac{e^{2\lambda}I_1+e^{-2\lambda}I_2}{\sqrt2},
\end{equation}
and substituting it in \eqref{eq:tgtbc}, with $\psi(1,t)$ and $\omega(1,t)$ written out as \eqref{eq:PsiOmega}, and solving for $u_x(1,t)$ gives \eqref{eq:rapidfb}. For well-posedness and the estimate, put $\psi=e^{-\lambda t}\varphi$; by \eqref{eq:psiomega} and \eqref{eq:tgtbc}, $\varphi$ is the chain with a spring,
\begin{equation}\label{eq:phispring}
\varphi_{tt}=(x\,\varphi_x)_x,\qquad \lim_{x\to0^+}x\,\varphi_x=0,\qquad \varphi_x(1,t)=-c_0\,\varphi(1,t),
\end{equation}
whose generator is skew-adjoint by Lemma \ref{lem:skew} and which therefore conserves $\int_0^1(x\varphi_x^2+\varphi_t^2)dx+c_0\varphi^2(1)$. Hence the conservative chain with a spring generates a unitary group, multiplication by $e^{-\lambda t}$ gives the target semigroup, and $N_\psi(t)=e^{-\lambda t}N_\psi(0)$ exactly. The augmented similarity of \eqref{eq:normequiv}, which Lemmas \ref{lem:HS}, \ref{lem:shift} and \ref{lem:trace} establish, is bounded with bounded inverse and transports that semigroup to the closed loop, which gives existence and uniqueness and, with the exact decay, \eqref{eq:ratedecay}. \hfill$\square$

\begin{remark}\label{rem:impedance}
The collocated part of \eqref{eq:rapidfb} is velocity feedback at the actuated end with gain $\tanh(2\lambda)$, a mismatched impedance; the matched damper $u_x(1,t)=-u_t(1,t)$ is its limit as $\lambda\to\infty$. The line so placed is the one relocated in \cite{SmyshlyaevKrstic} for the anti-damped string, where the reflection at the uncontrolled end fixes it at $\frac12\ln\frac{1+q}{1-q}$ and the design moves it to $-\frac12\ln\frac{1+c}{1-c}$; the factor $\frac14$ here in place of $\frac12$ is the round trip $4$ in place of $2$. The constant in \eqref{eq:ratedecay} cannot be uniform in $\lambda$: a state whose energy travels away from the trolley is unaffected until information crosses the domain, which forces $M_\lambda\gtrsim e^{2\lambda}$.
\end{remark}

\subsection{Transformation by Abel transmutation}\label{sec:explicit}

Equations \eqref{eq:transf} present the transformation through its kernels. The same transformation has a second form in which no kernel appears, and which shows where the kernels come from: it is the elementary transformation of the ordinary string, seen through the map that turns the string into the chain.

That map is classical. With
\begin{equation}\label{eq:PO}
(\mathcal P_0f)(r)=\frac2\pi\int_0^r\frac{f(y)}{\sqrt{r^2-y^2}}\,dy,\qquad
(\mathcal P_1f)(r)=\frac{2}{\pi r}\int_0^r\frac{y\,f(y)}{\sqrt{r^2-y^2}}\,dy,
\end{equation}
whose inverses are
\begin{equation}\label{eq:POinv}
(\mathcal P_0^{-1}g)(r)=\frac{d}{dr}\int_0^r\frac{y\,g(y)}{\sqrt{r^2-y^2}}\,dy,\qquad
(\mathcal P_1^{-1}g)(r)=\frac1r\frac{d}{dr}\int_0^r\frac{y^2g(y)}{\sqrt{r^2-y^2}}\,dy,
\end{equation}
if $v$ solves the ordinary string $v_{tt}=v_{rr}$ with $v_r(0,t)=0$ and $u=\mathcal P_0v$, then $u$ solves \eqref{eq:plant}--\eqref{eq:tip} in the coordinate $r=\xh$, and the velocity and strain of the two are related by
\begin{equation}\label{eq:PPderiv}
u_t=\mathcal P_0v_t,\qquad u_r=\mathcal P_1v_r,\qquad
\mathcal P_0\cosh(s\,\cdot)=I_0(s\,\cdot),\qquad \mathcal P_1\sinh(s\,\cdot)=I_1(s\,\cdot).
\end{equation}
On the string, uniform damping is produced by multiplying the two invariants $v_t\pm v_r$ by $e^{\pm\lambda r}$, which on the velocity and strain is the pair of hyperbolic combinations below.

\begin{proposition}[Explicit transformation]\label{prop:explicit}
Let
\begin{subequations}\label{eq:expl}
\begin{align}
\gamma&=\mathcal P_0\Bigl[\cosh(\lambda r)\,\mathcal P_0^{-1}u_t+\sinh(\lambda r)\,\mathcal P_1^{-1}u_r\Bigr],\label{eq:explA}\\
\delta&=\mathcal P_1\Bigl[\sinh(\lambda r)\,\mathcal P_0^{-1}u_t+\cosh(\lambda r)\,\mathcal P_1^{-1}u_r\Bigr].\label{eq:explB}
\end{align}
\end{subequations}
Then the transformation \eqref{eq:transf} is
\begin{equation}\label{eq:VWtoab}
\alpha=\frac{\sqrt r}{2}\,\bigl(\gamma+\delta\bigr),\qquad \beta=\frac{\sqrt r}{2}\,\bigl(\delta-\gamma\bigr).
\end{equation}
Writing $\Psi_s(r)=\frac{\sqrt r}{2}\bigl(I_0(sr)+I_1(sr),\ I_1(sr)-I_0(sr)\bigr)$, so that $A\Psi_s=s\Psi_s$,
\begin{equation}\label{eq:shift}
\mathcal T_\lambda\Psi_s=\Psi_{s+\lambda},\qquad
\mathcal T_\mu\mathcal T_\lambda=\mathcal T_{\lambda+\mu},\qquad
\mathcal T_\lambda^{-1}=\mathcal T_{-\lambda}.
\end{equation}
\end{proposition}

So $\gamma$ and $\delta$ are the velocity and the strain of the transformed state, and \eqref{eq:VWtoab} is the same combination that makes $\xib,\etab$ out of $u_t,u_r$ in \eqref{eq:riemanndef}. Nothing here is a further change of variables: \eqref{eq:expl} and \eqref{eq:transf} are one map written two ways, the first without kernels, so $\mathcal T_\lambda$ can be applied by quadrature.

It might appear from \eqref{eq:explA}--\eqref{eq:VWtoab} that backstepping was not even needed in order to arrive at the transformation to a desirable target. However, \eqref{eq:explA}--\eqref{eq:VWtoab} are less useful than they appear. They are not a control law: the force sits inside $\mathcal P_1^{-1}$, so it is implicit, and the Volterra form yields \eqref{eq:rapidfb}. And they are available only because this plant is the radial wave equation exactly; the kernel route does not depend on that coincidence, and is a design direction for coupled hyperbolic systems more broadly. What \eqref{eq:explA}--\eqref{eq:VWtoab} do give is the reading of the kernels: they are not an artifact of the method but the coordinate expression of an elementary similarity, which is what tells us the design is right rather than merely verified.

\emph{Proof of Proposition \ref{prop:explicit}.} On the string, multiplying the invariants $v_t\pm v_r$ by $e^{\pm\lambda r}$ produces the same string with $-\lambda$ added to both equations, and in the variables $(v_t,v_r)$ that multiplication is the matrix with entries $\cosh\lambda r$ on the diagonal and $\sinh\lambda r$ off it. Conjugating by $\mathcal P_0$ and $\mathcal P_1$ through \eqref{eq:PPderiv} gives \eqref{eq:expl}, and \eqref{eq:VWtoab} is \eqref{eq:riemanndef} applied to the transformed velocity and strain. The map so defined is identity-leading, triangular because $\mathcal P_0,\mathcal P_1$ and their inverses are, and it intertwines $A$ with $A-\lambda I$; by the uniqueness in Proposition \ref{prop:existence} its kernels are those of \eqref{eq:Ksystem}--\eqref{eq:Kdiag}, so it is $\mathcal T_\lambda$ and \eqref{eq:VWtoab} holds. For \eqref{eq:shift}, the state $\Psi_s$ has $u_t=I_0(sr)$ and $u_r=I_1(sr)$, so by \eqref{eq:PPderiv} its pullback is $\bigl(\cosh(sr),\sinh(sr)\bigr)$; the matrix adds $\lambda$ to the argument, and pushing forward gives $\bigl(I_0((s+\lambda)r),I_1((s+\lambda)r)\bigr)=\Psi_{s+\lambda}$. The group property is that of the matrix, and the inverse follows. \hfill$\square$

The factors in \eqref{eq:expl} are unbounded, $\mathcal P_0$ and $\mathcal P_1$ smoothing by half a derivative and their inverses losing it, but the half orders cancel in the conjugation: the composition is the bounded operator of Lemma \ref{lem:HS}, with the kernel bound \eqref{eq:Kbound}. 

So the backstepping transformation is a translation of the radial spectral parameter, $s\mapsto s+\lambda$, which is why the closed loop is similar to $A-\lambda I$ and not asymptotically similar to it, and why the construction is not a perturbation of pole placement.

\subsection{Trolley displacement and the augmented norm}\label{sec:displacement}

Two elementary bounds are used throughout. From $u(x)=u(1)-\int_x^1u_y$ and the weighted Cauchy--Schwarz inequality, $|u(x)-u(1)|^2\le\ln\frac1x\int_x^1y\,u_y^2\,dy$, and $\int_0^1\ln\frac1x\,dx=1$, so
\begin{equation}\label{eq:poincare}
\|u\|^2\le2u^2(1)+2\bigl\|\sqrt x\,u_x\bigr\|^2 ,
\end{equation}
and the same holds for $\psi$. From $w=x^{-1/2}\int_0^xu_t$ and Cauchy--Schwarz, $|w(x)|\le\|u_t\|_{L^2(0,x)}$ at every point, so $\|w\|\le\|u_t\|$, and the same bound holds for the companion of $\psi$.

The pair carried by $N_\psi$ is not the image of the state. Applying \eqref{eq:transf} to the state gives $(\omega_t,\psi_t)$, since $(w,u)_t=(\sqrt x\,u_x,u_t)$ and the transformation is time-invariant, whereas $N_\psi$ is built on $\bigl(\psi_t+\lambda\psi,\ \sqrt x\,\psi_x\bigr)$. By linearity and \eqref{eq:psiomega} the two are related by one shift:
\begin{equation}\label{eq:bridge}
\mathcal T_\lambda\bigl(\sqrt x\,u_x+\lambda w,\ u_t+\lambda u\bigr)=\bigl(\omega_t+\lambda\omega,\ \psi_t+\lambda\psi\bigr)=\bigl(\sqrt x\,\psi_x,\ \psi_t+\lambda\psi\bigr),
\end{equation}
so Lemma \ref{lem:HS} applies verbatim to the shifted state, with the same kernels and the same bound. The shifted and unshifted plant norms are not equivalent on their own, and the trolley displacement is what makes them so.

\begin{lemma}[Shifted state]\label{lem:shift}
With $N_u$ the norm \eqref{eq:normc0},
\begin{equation}\label{eq:shiftequiv}
C_{1}\,N_u^2\ \le\ \bigl\|\sqrt x\,u_x+\lambda w\bigr\|^2+\bigl\|u_t+\lambda u\bigr\|^2+c_0u^2(1)\ \le\ C_{2}\,N_u^2 .
\end{equation}
\end{lemma}

\emph{Proof of Lemma \ref{lem:shift}.} Write $\mathcal B_0(w,u)=\bigl(\sqrt x\,u_x,\ u_t,\ \sqrt{c_0}\,u(1)\bigr)$ and $\mathcal B_\lambda=\mathcal B_0+\lambda\mathcal C$ with $\mathcal C(w,u)=(w,u,0)$, so that $\|\mathcal B_0\|^2=N_u^2$ and $\|\mathcal B_\lambda\|^2$ is the middle member of \eqref{eq:shiftequiv}; the upper bound is then $\|u\|\le C(\|\sqrt x\,u_x\|+|u(1)|)$ from \eqref{eq:poincare} together with $\|w\|\le\|u_t\|$.

For the lower bound, $\mathcal B_0$ is boundedly invertible, and while its inverse is not compact, the composition $\mathcal C\mathcal B_0^{-1}$ is: its two $L^2$-valued components are Hilbert--Schmidt integral operators and its trace component has finite rank. Indeed, write $(f,g,a)=\mathcal B_0(w,u)$. Then $u_x=f/\sqrt x$ integrated from the trolley, and $\bigl(\sqrt x\,w\bigr)_x=g$ with $\sqrt x\,w$ vanishing at the tip, so
\begin{equation}\label{eq:B0inv}
u(x)=\frac{a}{\sqrt{c_0}}-\int_x^1\frac{f(s)}{\sqrt s}\,ds,\qquad
w(x)=\frac{1}{\sqrt x}\int_0^x g(s)\,ds,
\end{equation}
and the two integral operators are Hilbert--Schmidt, their squared norms both equal to $1$, while the term carrying $u(1)$ has rank one. Hence $\mathcal B_\lambda\mathcal B_0^{-1}=I+\lambda\,\mathcal C\mathcal B_0^{-1}$ is a compact perturbation of the identity, Fredholm of index zero, and it suffices to show that it is injective.

Let $\mathcal B_\lambda(w,u)=0$. In the transit coordinate $r=2\sqrt x$, where $\sqrt x\,\partial_x=\partial_r$, the first two components read $u_r+\lambda w=0$ and $w_r+w/r+\lambda u=0$, so $w=-u_r/\lambda$ and
\begin{equation}\label{eq:shiftker}
u_{rr}+\frac1r\,u_r-\lambda^2u=0 ,
\end{equation}
whose solutions are $C_1I_0(\lambda r)+C_2K_0(\lambda r)$. The second is excluded because $K_0(\lambda r)\sim-\ln r$ gives $u_r\sim-1/r$ and $\int_0 u_r^2\,dx=\int_0 r^{-1}dr=\infty$, outside the energy space. The third component gives $u(1)=0$, that is $C_1I_0(2\lambda)=0$, and $I_0(2\lambda)>0$, so $C_1=0$ and $w=0$. \hfill$\square$

The kernel found in that proof is the reason the trolley displacement cannot be dispensed with: without its third component $\mathcal B_\lambda$ annihilates $u=C\,I_0(2\lambda\sqrt x)$, and the spring coordinate removes exactly that direction. What remains is to exchange the two trolley traces.

\begin{lemma}[Trolley traces]\label{lem:trace}
With $N_\psi^2=\bigl\|\sqrt x\,\psi_x\bigr\|^2+\bigl\|\psi_t+\lambda\psi\bigr\|^2+c_0\psi^2(1)$ and $\|\cdot\|$ the norm of $L^2(0,1)$,
\begin{equation}\label{eq:tracefwd}
|\psi(1)|\le C_\lambda\Bigl(\bigl\|\sqrt x\,u_x\bigr\|+\|u_t\|+|u(1)|\Bigr),\qquad
|u(1)|\le C_\lambda\,N_\psi .
\end{equation}
\end{lemma}

\emph{Proof of Lemma \ref{lem:trace}.} Evaluating \eqref{eq:transf}, applied to $(w,u)$, at the trolley gives
\begin{equation}\label{eq:trolley}
\psi(1)=\cosh(2\lambda)\,u(1)+\sinh(2\lambda)\,w(1)
+\frac{e^{-2\lambda}I_2^d-e^{2\lambda}I_1^d}{\sqrt2},
\end{equation}
with $J_1,J_2$ as in \eqref{eq:Idef}; that is, $\psi(1,t)=\Psi(t)$ and likewise $\omega(1,t)=\Omega(t)$ of \eqref{eq:PsiOmega}. Three bounds close it. First, $w=p/\sqrt x$ with $p(x)=\int_0^xu_t$, so $|p(x)|\le\sqrt x\,\|u_t\|$ and hence $|w(x)|\le\|u_t\|$ for every $x$, giving both $|w(1)|\le\|u_t\|$ and $\|w\|\le\|u_t\|$. Second, \eqref{eq:poincare}. Third, the integrals in \eqref{eq:trolley} are bounded by $C_\lambda(\|u\|+\|w\|)$, the kernels being Hilbert--Schmidt by \eqref{eq:HScalc}. Together these give the first estimate.

For the second, apply \eqref{eq:trolley} at $-\lambda$, which is the inverse by Lemma \ref{lem:HS}, obtaining $u(1)$ from $\psi(1)$, from the companion $W$ of the target and from integrals of $\psi$ and $W$. Since $\psi_t+\lambda\psi=\bigl(\sqrt x\,W\bigr)_x$ and $\sqrt x\,W$ vanishes at the tip, $W(1)=\int_0^1(\psi_t+\lambda\psi)\,dx$, so $|W(1)|\le\|\psi_t+\lambda\psi\|$ and likewise $\|W\|\le\|\psi_t+\lambda\psi\|$; and \eqref{eq:poincare} applied to $\psi$ bounds $\|\psi\|$ by $\|\sqrt x\,\psi_x\|$ and $|\psi(1)|$. Each of these reaches a different member of $N_\psi$, so the estimate closes for every $\lambda>0$, with no smallness restriction. It is here that the shifted velocity is needed: with $\psi_t$ in place of $\psi_t+\lambda\psi$ the bounds return to their starting point and close only for small $\lambda$. \hfill$\square$

Lemmas \ref{lem:HS}, \ref{lem:shift} and \ref{lem:trace} together give
\begin{equation}\label{eq:normequiv}
N_u\ \asymp\ \bigl\|\sqrt x\,u_x+\lambda w\bigr\|+\bigl\|u_t+\lambda u\bigr\|+\sqrt{c_0}\,|u(1)|\ \asymp\ N_\psi ,
\end{equation}
and $N_\psi(t)=e^{-\lambda t}N_\psi(0)$ from $\psi=e^{-\lambda t}\varphi$ with $\varphi$ conserved by Lemma \ref{lem:skew}, which is \eqref{eq:ratedecay}. By \eqref{eq:poincare} the displacement itself decays with the state.

Two quantities in this construction decay at the rate $\lambda$, and they are not one functional in two coordinate systems. The first is \eqref{eq:tgtdecay}, in the Riemann variables of the target and its trace $\sigma$; the second is $N_\psi$, in the target displacement. The estimate \eqref{eq:ratedecay} rests on the second, through \eqref{eq:normequiv}.

\section{Kernels}\label{sec:kernels}

\subsection{Kernel well-posedness and transform invertibility}\label{sec:rapidkernels}

The characteristics of \eqref{eq:Kxx} and \eqref{eq:Kyy} are the lines $\xh-\yh=\text{const}$, which enter the triangle at the edge $\yh=0$ and run parallel to the diagonal, the diagonal being one of them; those of \eqref{eq:Kxy} and \eqref{eq:Kyx} are the lines $\xh+\yh=\text{const}$, which enter at the diagonal and reach the edge, so \eqref{eq:Kdiag} is the data for the second pair and the first pair is closed at the edge, where the coefficients are singular. Two facts govern that edge, and both follow from \eqref{eq:tilt}.

\begin{theorem}[Kernel well-posedness]\label{prop:existence}
The problem \eqref{eq:Ksystem}, \eqref{eq:Kdiag}, \eqref{eq:branchcond} has a unique solution on $0<\yh<\xh<2$, and it satisfies
\begin{equation}\label{eq:Kbound}
\bigl|K^{\xi\xi}\bigr|+\bigl|K^{\xi\eta}\bigr|+\bigl|K^{\eta\xi}\bigr|+\bigl|K^{\eta\eta}\bigr|\le C_\lambda\sqrt{\frac{\yh}{\xh}} .
\end{equation}
\end{theorem}
This is the well-posedness assumed in Theorem \ref{thm:rapid}. Its proof is given after Lemma \ref{lem:cancel}, which supplies the property of the diagonal that the proof needs.

\begin{lemma}[Diagonal cancellation]\label{lem:cancel}
The right sides of \eqref{eq:Kxx} and \eqref{eq:Kyy}, evaluated on the diagonal with the data \eqref{eq:Kdiag}, vanish identically in $\xh$ and $\lambda$. Consequently $K^{\xi\xi}$ and $K^{\eta\eta}$ are constant along the diagonal, and
\begin{equation}\label{eq:diagconst}
K^{\xi\xi}(\xh,\xh)=\frac{\lambda}{2},\qquad K^{\eta\eta}(\xh,\xh)=-\frac{\lambda}{2}.
\end{equation}
\end{lemma}
\emph{Proof.} With \eqref{eq:Kdiag}, the right side of \eqref{eq:Kxx} on $\yh=\xh$ is
\begin{equation}\label{eq:cancelcalc}
-\frac{1}{2\xh}\cdot\frac{e^{-2\lambda\xh}-1}{4\xh}-\frac{e^{-2\lambda\xh}}{2\xh}\cdot\frac{e^{2\lambda\xh}-1}{4\xh}
=-\frac{1}{8\xh^2}\Bigl[\bigl(e^{-2\lambda\xh}-1\bigr)+\bigl(1-e^{-2\lambda\xh}\bigr)\Bigr]=0,
\end{equation}
and the right side of \eqref{eq:Kyy} on the diagonal vanishes by the same computation with $e^{2\lambda\xh}$ in place of $e^{-2\lambda\xh}$. The diagonal is a characteristic of \eqref{eq:Kxx} and of \eqref{eq:Kyy}, whose transport direction is $\partial_{\xh}+\partial_{\yh}$, so $\frac{d}{d\xh}K^{\xi\xi}(\xh,\xh)=0$ and $\frac{d}{d\xh}K^{\eta\eta}(\xh,\xh)=0$; evaluating \eqref{eq:Kleadb} at $\sigma=1$ gives the constants \eqref{eq:diagconst}. \hfill$\square$

The diagonal is itself a characteristic of \eqref{eq:Kxx}, so without this cancellation the transport of $K^{\xi\xi}$ along it would carry $\int_0 d\yh/\yh$ and diverge at the corner.

Nothing is prescribed at $\yh=0$, and nothing can be: the coefficients $1/2\yh$ are infinite there, and near that end the solutions of \eqref{eq:Ksystem} are combinations of $\yh^{1/2}$ and $\yh^{-1/2}$. Requiring the kernels to be bounded discards the second, and that requirement closes the problem in place of a boundary condition. On the bounded solutions, with $\sigma=\yh/\xh$,
\begin{equation}\label{eq:Kleadb}
K^{\xi\xi}=K^{\eta\xi}=\ \ \frac{\lambda}{2}\sqrt\sigma+O(\lambda^2\xh),\qquad
K^{\xi\eta}=K^{\eta\eta}=-\frac{\lambda}{2}\sqrt\sigma+O(\lambda^2\xh),
\end{equation}

The scaled kernels are the natural unknowns. Writing
\begin{equation}\label{eq:scaledk}
K^{\xi\xi}=\frac{\sqrt\sigma}{\xh}\,k^{\xi\xi}(z,\sigma),\qquad z=\lambda\xh,\quad\sigma=\frac{\yh}{\xh},
\end{equation}
and likewise for the other three, boundedness at $\yh=0$ is the pair of conditions
\begin{equation}\label{eq:branchcond}
k^{\xi\xi}(z,0)+k^{\xi\eta}(z,0)=0,\qquad k^{\eta\xi}(z,0)+k^{\eta\eta}(z,0)=0,
\end{equation}
since it is the sum in each pair that multiplies $\yh^{-1/2}$. 

\emph{Proof of Theorem \ref{prop:existence}.} Expand $k=\sum_{n\ge1}z^nk_n(\sigma)$ for each of the four scaled kernels \eqref{eq:scaledk}. At order $n$ the equations reduce to a two-point problem in $\sigma$ whose homogeneous equation is hypergeometric with terminating solutions, so $k_n$ is a polynomial of degree at most $n-1$, determined uniquely by regularity at $\sigma=0$ and by \eqref{eq:Kdiag}, and the derivative $z\partial_z$ contributes $n$ to its coefficient. Measuring in the weight of \eqref{eq:Edef} and using the cancellation \eqref{eq:sigcancel} at each order gives
\begin{equation}\label{eq:coeffrec}
a_n:=\|k_n\|_w\le\frac{C}{n}\sum_{j=1}^{n-1}\frac{2^j}{j!}\,a_{n-j}+C\,\frac{2^n}{n!},\qquad n\ge n_0,
\end{equation}
the sum coming from the Taylor coefficients of $e^{\pm2z}$ in \eqref{eq:Ksystem} and the last term from \eqref{eq:Kdiag}. Fix $R>0$ and suppose $a_mR^m\le M_R$ for $m<n$. Then \eqref{eq:coeffrec} gives $a_nR^n\le C(e^{2R}-1)M_R/n+Ce^{2R}$, so for $n\ge2C(e^{2R}-1)$ the first term is at most $M_R/2$, and enlarging $M_R$ to cover the finitely many earlier coefficients gives $a_n\le M_RR^{-n}$ for every $n$. Since $R$ is arbitrary, the series has infinite radius in $z$, and in particular converges on $0\le z\le2\lambda$.

Written out, with $a_n,b_n,c_n,d_n$ the coefficients of the four scaled kernels \eqref{eq:scaledk}, the order-$n$ problem is
\begin{subequations}\label{eq:Rn}
\begin{eqnarray}
(1-\sigma)a_n'+\bigl(n-\tfrac32\bigr)a_n+\frac{a_n+b_n}{2\sigma}+\tfrac12c_n
&=&-\tfrac12\sum_{j=1}^{n-1}\frac{(-2)^j}{j!}\,c_{n-j},\label{eq:Rna}\\
-(1+\sigma)b_n'+\bigl(n-\tfrac32\bigr)b_n-\frac{a_n+b_n}{2\sigma}+\tfrac12d_n
&=&-\tfrac12\sum_{j=1}^{n-1}\frac{(-2)^j}{j!}\,d_{n-j},\label{eq:Rnb}\\
-(1+\sigma)c_n'+\bigl(n-\tfrac32\bigr)c_n-\frac{c_n+d_n}{2\sigma}+\tfrac12a_n
&=&-\tfrac12\sum_{j=1}^{n-1}\frac{2^j}{j!}\,a_{n-j},\label{eq:Rnc}\\
(1-\sigma)d_n'+\bigl(n-\tfrac32\bigr)d_n+\frac{c_n+d_n}{2\sigma}+\tfrac12b_n
&=&-\tfrac12\sum_{j=1}^{n-1}\frac{2^j}{j!}\,b_{n-j},\label{eq:Rnd}
\end{eqnarray}
\end{subequations}
with $a_n(0)+b_n(0)=0$ and $c_n(0)+d_n(0)=0$ from \eqref{eq:branchcond}, and
\begin{equation}\label{eq:Rndiag}
b_n(1)=\frac{(-2)^n}{4\,n!},\qquad c_n(1)=\frac{2^n}{4\,n!}
\end{equation}
from \eqref{eq:Kdiag}. At $n=1$ the sums are empty and \eqref{eq:Rn}--\eqref{eq:Rndiag} give the constants $a_1=c_1=\tfrac12$, $b_1=d_1=-\tfrac12$ of \eqref{eq:Klead}; each further order is one linear two-point problem in $\sigma$ with the preceding orders as data.

The singular terms of \eqref{eq:Ksystem} converge with it. For a polynomial of degree at most $n$ on $(0,1)$, $\|p\|_\infty\le C(n+1)^2\|p\|_w$ and $\|p'\|_\infty\le Cn^2\|p\|_\infty$, so $\|k_n\|_\infty\le Cn^2M_RR^{-n}$; and since branch selection gives $\bigl(k_n^{\xi\xi}+k_n^{\xi\eta}\bigr)(0)=0$, the mean value theorem and the second inequality give
\begin{equation}\label{eq:divdiff}
\Bigl\|\frac{k_n^{\xi\xi}+k_n^{\xi\eta}}{\sigma}\Bigr\|_\infty\le Cn^4M_RR^{-n},
\end{equation}
and the same for the other pair. The polynomial factors do not change the radius, so those series converge uniformly on $0\le z\le2\lambda$ as well, and the sum solves \eqref{eq:Ksystem} including its singular terms, with \eqref{eq:branchcond} and \eqref{eq:Kdiag}. Because the series begins at $n=1$, $k=O(z)$ uniformly, and \eqref{eq:scaledk} gives \eqref{eq:Kbound}, the factor $z=\lambda\xh$ cancelling the $1/\xh$. Uniqueness follows from \eqref{eq:Egron} applied to the difference of two solutions with vanishing data. \hfill$\square$

The weight $\sigma$ removes the singular terms from the energy balance, which is what gives uniqueness: the two terms carrying $1/2\sigma$ combine into
\begin{equation}\label{eq:sigcancel}
-\frac12\int_0^1\bigl(k^{\xi\xi}+k^{\xi\eta}\bigr)\bigl(k^{\xi\xi}-k^{\xi\eta}\bigr)d\sigma
\quad\text{and}\quad
\frac12\int_0^1\bigl(k^{\eta\xi}+k^{\eta\eta}\bigr)\bigl(k^{\eta\xi}-k^{\eta\eta}\bigr)d\sigma,
\end{equation}
in which no $1/\sigma$ survives, where
\begin{equation}\label{eq:Edef}
E(\rho)=\frac12\int_0^1\sigma\Bigl(|k^{\xi\xi}|^2+|k^{\xi\eta}|^2+|k^{\eta\xi}|^2+|k^{\eta\eta}|^2\Bigr)d\sigma,
\end{equation}
the transport terms carrying $\sigma(1\mp\sigma)$, which vanishes at $\sigma=0$ in both cases and at $\sigma=1$ only for the first, so that the traces of the other pair remain and
\begin{equation}\label{eq:Egron}
E'(\rho)\le\bigl[1+\cosh 2z\bigr]E(\rho)+\bigl|k^{\xi\eta}(\rho,1)\bigr|^2+\bigl|k^{\eta\xi}(\rho,1)\bigr|^2 .
\end{equation}
Existence is by the expansion in powers of $z$, whose coefficients are polynomials in $\sigma$ and whose radius is infinite, so no continuation in $z$ is needed.

\begin{lemma}[Transformation]\label{lem:HS}
Under \eqref{eq:Kbound} the Volterra part $\mathcal K_\lambda$ of \eqref{eq:transf} is Hilbert--Schmidt on $L^2(0,2)^2$, which by \eqref{eq:isometry} is the energy space, and the multiplier $M_\lambda$ is bounded with bounded inverse there, so $\mathcal T_\lambda=M_\lambda\bigl(I-\mathcal K_\lambda\bigr)$ is bounded. Writing $\mathcal T_\lambda$ for the whole of \eqref{eq:transf},
\begin{equation}\label{eq:inverse}
\mathcal T_\lambda^{-1}=\mathcal T_{-\lambda},
\end{equation}
so $\mathcal T_\lambda$ is boundedly invertible on the energy space and $M_\lambda=\|\mathcal T_\lambda\|\,\|\mathcal T_{-\lambda}\|$.
\end{lemma}

\emph{Proof.} By \eqref{eq:Kbound},
\begin{equation}\label{eq:HScalc}
\int_0^2\!\!\int_0^{\xh}\bigl|K(\xh,\yh)\bigr|^2d\yh\,d\xh\le C_\lambda^2\int_0^2\!\!\int_0^{\xh}\frac{\yh}{\xh}\,d\yh\,d\xh=C_\lambda^2\int_0^2\frac{\xh}{2}\,d\xh<\infty
\end{equation}
for each of the four kernels. For \eqref{eq:inverse}, let $A$ denote the spatial operator of \eqref{eq:sys}. Theorem \ref{thm:rapid} is the statement $\mathcal T_\lambda A=(A-\lambda I)\mathcal T_\lambda$, and the same construction with $-\lambda$ gives $\mathcal T_{-\lambda}(A-\lambda I)=A\,\mathcal T_{-\lambda}$, so $S=\mathcal T_{-\lambda}\mathcal T_\lambda$ commutes with $A$. Since $\bigl(I-\mathcal K_{-\lambda}\bigr)M_\lambda=M_\lambda\bigl(I-M_\lambda^{-1}\mathcal K_{-\lambda}M_\lambda\bigr)$ and $M_{-\lambda}M_\lambda=I$, $S$ is again of the form $I$ minus a Volterra operator, and it maps the plant to itself, so its kernels solve \eqref{eq:Ksystem} with $\lambda=0$, for which \eqref{eq:Kdiag} vanishes; by the uniqueness in Proposition \ref{prop:existence} those kernels are zero and $S=I$. \hfill$\square$

The transformation acts on the energy space itself, with no weight and no loss of derivatives.

\subsection{All kernels obtained from a single scalar function}\label{sec:onekernel}

The four kernels are one. Two of them are the first two with the rate reversed,
\begin{equation}\label{eq:symmlam}
K^{\eta\eta}_\lambda=K^{\xi\xi}_{-\lambda},\qquad K^{\eta\xi}_\lambda=K^{\xi\eta}_{-\lambda},
\end{equation}
as $\mathcal T_\lambda^{-1}=\mathcal T_{-\lambda}$ requires. The remaining relation is a reflection in the second argument, which leaves the triangle and is therefore written on the scaled kernel. Putting
\begin{equation}\label{eq:scalara}
K^{\xi\xi}(\xh,\yh)=\frac{\sqrt\sigma}{\xh}\,a(z,\sigma),\qquad z=\lambda\xh,\quad\sigma=\frac{\yh}{\xh},
\end{equation}
one has
\begin{equation}\label{eq:symm}
K^{\xi\eta}=-\frac{\sqrt\sigma}{\xh}\,a(z,-\sigma),
\end{equation}
so that all four follow from $a$ alone. Its series begins
\begin{equation}\label{eq:aseries}
a(z,\sigma)=\frac{z}{2}+\frac{z^2}{4}(\sigma-1)+\frac{z^3}{24}(\sigma-1)(3\sigma-1)+O(z^4),
\end{equation}
its value $z/2$ at $\sigma=1$ being the constancy \eqref{eq:diagconst} on the diagonal, and it is entire in $z$ and polynomial in $\sigma$ at each order.

The gains inherit the reduction. Writing $a_e$ and $a_o$ for the even and odd parts in $\sigma$ of the scalar kernel $a$ of \eqref{eq:scalara}, the gains of \eqref{eq:rapidfb} are
\begin{eqnarray}
g^\lambda_1(x)&=&\frac{e^{4\lambda}\,a_e(2\lambda,\sqrt x)+a_e(-2\lambda,\sqrt x)}{e^{4\lambda}-1},\label{eq:gain1}\\
g^\lambda_2(x)&=&\sqrt x\;\;\frac{e^{4\lambda}\,a_o(2\lambda,\sqrt x)-a_o(-2\lambda,\sqrt x)}{e^{4\lambda}-1}:\label{eq:gain2}
\end{eqnarray}
the velocity gain is the even part of one scalar kernel, the force gain its odd part. There are not four independent gain kernels, and no kernel equation has to be solved to evaluate them, since \eqref{eq:expl} determines $a$ by quadrature.

\section{Simulation}\label{sec:sim}

The chain is released from rest at a high mode, $u(x,0)$ built from $J_0(j_{0,6}\sqrt x)$ with a little of the next one and with the $x^{-1/4}$ growth of the envelope removed, so that the excursions are comparable along the length and the zeros crowd toward the tip, as in Figure \ref{fig:chain}. A rigid offset is added, so that the initial shape is displaced as well as curved, and a multiple of $x^8$ is added to make the initial slope at the trolley agree with the force that \eqref{eq:rapidfb} calls for at $t=0$; without that agreement the mismatch travels to the free end, is amplified there by the focusing of Remark \ref{rem:focusing}, and returns after one round trip.

Figure \ref{fig:snapshots} shows the response without control and under \eqref{eq:rapidfb} at $\lambda=1$ and $c_0=1$. The uncontrolled chain conserves its energy and keeps swinging about the offset; the controlled one is on the vertical $u=0$ within two transit times. The offset is the point of the run: with the spring removed from the target, that is with $\alpha(2,t)=\beta(2,t)$ in place of \eqref{eq:tgtbc}, the same simulation straightens the chain but leaves $\max_x|u|=0.61$ standing for all time, while \eqref{eq:rapidfb} brings it to $3\times10^{-4}$.

The kernels are those of \eqref{eq:Kseries}, with the coefficients generated by the recursion \eqref{eq:Rn}--\eqref{eq:Rndiag} and the series truncated at order twelve; the gains \eqref{eq:gain1}--\eqref{eq:gain2} of the distributed part are shown in Figure \ref{fig:gains}, and the force itself is Figure \ref{fig:force}; it oscillates at the frequency of the mode being cancelled and is spent within one round trip, after which it decays with the state. The kernels were verified against Proposition \ref{prop:explicit}, the transformation carrying $\Psi_s$ to $\Psi_{s+\lambda}$ to nine digits.

\begin{figure}[t]
\centering
\includegraphics[width=0.95\textwidth]{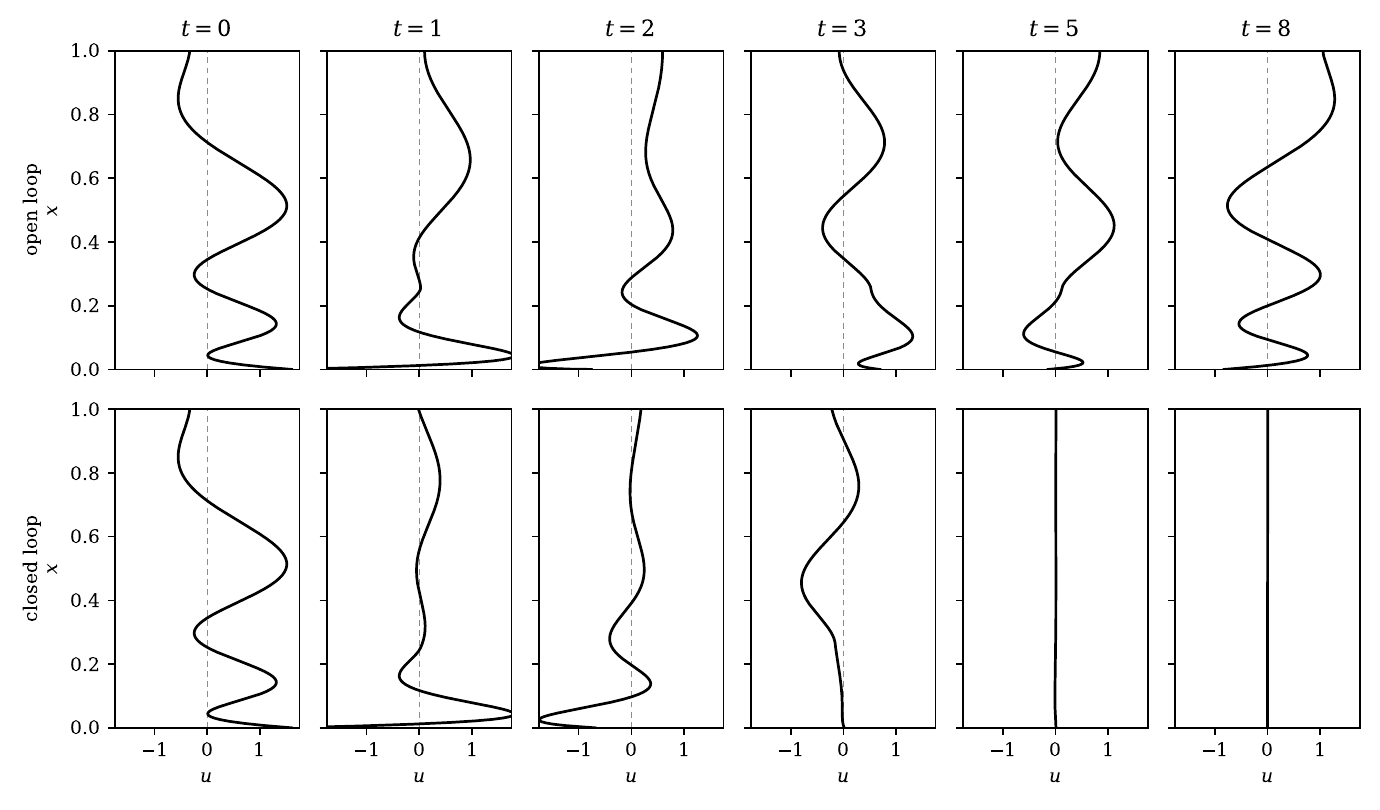}
\caption{The chain released from rest at a high mode displaced by a rigid offset: without control above, and under \eqref{eq:rapidfb} at $\lambda=1$, $c_0=1$ below. Height $x$ is measured from the free end, so the trolley is at the top of each panel and the tip at the bottom, and the vertical $u=0$ is dashed.}
\label{fig:snapshots}
\end{figure}

\begin{figure}[t]
\centering
\includegraphics[width=0.6\textwidth]{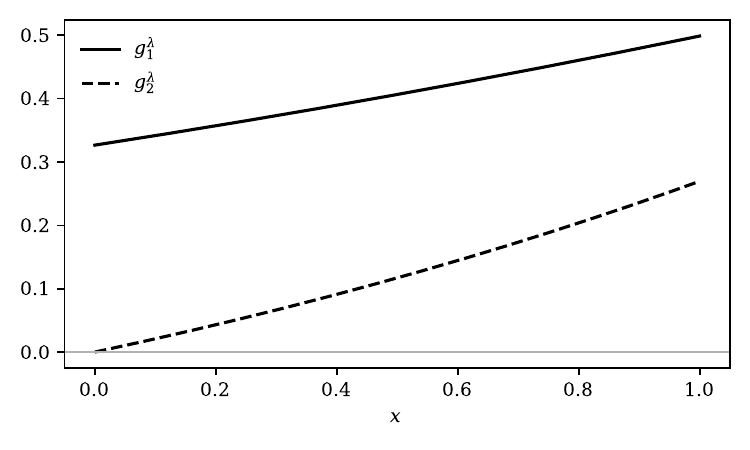}
\caption{The gains $g^\lambda_1$ and $g^\lambda_2$ of the distributed part of \eqref{eq:rapidfb} at $\lambda=1$. The strain gain vanishes at the free end, where the tension does.}
\label{fig:gains}
\end{figure}

\begin{figure}[t]
\centering
\includegraphics[width=0.6\textwidth]{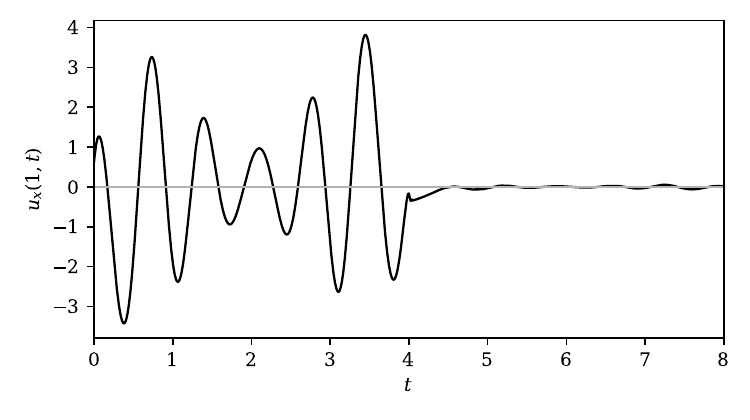}
\caption{The force at the trolley in the controlled run of Figure \ref{fig:snapshots}. The axis ends at $t=8$; beyond it the force continues to decay at the rate $\lambda$, its magnitude passing below $10^{-3}$.}
\label{fig:force}
\end{figure}

\section{Conclusions}

The chain is stabilized at any rate the designer names, by a feedback that is collocated in its local part and distributed in the rest, and the closed loop is the chain itself with damping along its length. The transformation that produces it is a Volterra map with an exponential factor, bounded with a bounded inverse in the energy of the chain, and it is the scaling of characteristics that damps an ordinary string, carried across the Abel transmutation.

Whether more is available was examined, and the answer is that it is not, along the three routes that suggest themselves. First, the singular coupling cannot be removed: no Volterra transformation of the second kind whose kernels are integrable against the behaviour at the free end carries the chain into a pair of uncoupled transport equations, and none carries it into a cascade in which one state is autonomous and drives the other. This is why the target of Section~\ref{sec:rapid} keeps the coupling rather than eliminating it, and it also settles the minimal time from the design side: an uncoupled target would bring the chain to rest in the transit time $2$, which the lower bound of \cite{PetitRouchon} forbids. Second, the minimal time $4$ is attainable, by transmuting the chain into an ordinary string and matching the impedance there, and the resulting closed loop is nilpotent. Third, that feedback is not admissible on the energy of the chain: its closed loop generates no bounded semigroup there, and the space on which it is nilpotent lies one half derivative above. We expect the same of every feedback that extinguishes the state in a finite time, of whatever length, since the attenuation between the free end and the trolley is one of frequency and waiting does not undo it.

There is a further reason not to pursue the finite time, and it is one of engineering rather than of function spaces. A feedback that extinguishes the state in a finite time must cancel the reflection at the actuated end, and for systems with strong reflection such cancellation destroys the delay margin altogether: the closed loop is unstable under an arbitrarily small delay in the loop, and the remedy is to leave part of the reflection in place and give up finite-time convergence \cite{Auriol}. The target of Section~\ref{sec:rapid} leaves all of it in place. What is renounced by keeping the reflection is the settling time; what is kept is a closed loop that is still a chain.

These results are established but not reported here. They answer questions rather than open them, and each requires an apparatus of its own; carrying them would double the length of the paper and divide its subject. What they leave is the statement this paper makes: for a plant whose degeneracy places it outside the reach of the existing designs, the reachable target is the chain with damping added, and it is reached.

\appendix

\section{Proof of Lemma \ref{lem:skew}}\label{sec:skewproof}

\begin{lemma}[Target generator]\label{lem:skew}
Let $V$ be the completion of $C^\infty[0,1]$ in $\|\phi\|_V^2=\int_0^1x\,\phi_x^2\,dx+c_0\phi^2(1)$, let $\mathcal H=V\times L^2(0,1)$ with $\|(\phi,\dot\phi)\|_{\mathcal H}^2=\|\phi\|_V^2+\|\dot\phi\|^2$, and let
\begin{equation}\label{eq:genA}
\mathcal A\binom{\phi}{\dot\phi}=\binom{\dot\phi}{-L\phi},\qquad D(\mathcal A)=D(L)\times V,
\end{equation}
where $L$ is the operator associated with the form $a(\phi,\chi)=\int_0^1x\,\phi_x\chi_x\,dx+c_0\phi(1)\chi(1)$. Then $\mathcal A$ is skew-adjoint and generates a unitary group on $\mathcal H$, and the solutions of \eqref{eq:phispring} conserve $\|\sqrt x\,\varphi_x\|^2+\|\varphi_t\|^2+c_0\varphi^2(1)$.
\end{lemma}

\emph{Proof of Lemma \ref{lem:skew}.} The form $a$ is symmetric, non-negative and closed on $V$, and $V$ embeds continuously in $L^2(0,1)$ by \eqref{eq:poincare}. By the representation theorem for closed forms it is associated with a positive self-adjoint operator $L$ on $L^2(0,1)$, $a(\phi,\chi)=\langle L\phi,\chi\rangle$. For $\phi\in D(L)$ smooth on $(0,1]$, integrating the form by parts gives
\begin{equation}\label{eq:formIBP}
a(\phi,\chi)=-\int_0^1\bigl(x\phi_x\bigr)_x\chi\,dx+\bigl[x\phi_x\chi\bigr]_0^1+c_0\phi(1)\chi(1),
\end{equation}
so $L\phi=-(x\phi_x)_x$. The tip condition is not assumed but follows: $(x\phi_x)_x\in L^2(0,1)$ puts $x\phi_x$ in $H^1(0,1)$, so it has a limit $\ell$ at $0$, and $\ell\ne0$ would give $\phi_x\sim\ell/x$ and $\int_0x\,\phi_x^2\,dx=\infty$, against $\phi\in V$; hence $\lim_{x\to0^+}x\phi_x=0$, which is \eqref{eq:tip}. With the tip term gone, \eqref{eq:formIBP} leaves $\bigl[\phi_x(1)+c_0\phi(1)\bigr]\chi(1)$, and $\chi(1)$ is arbitrary, so $\phi_x(1)=-c_0\phi(1)$: the spring. Because $L$ is positive and self-adjoint, $\mathcal A$ in \eqref{eq:genA} satisfies $\mathcal A^*=-\mathcal A$, and Stone's theorem gives the unitary group. Conservation of $\|\cdot\|_{\mathcal H}^2$, which is the stated quantity, follows. \hfill$\square$

\section*{Acknowledgments}

The author thanks Rafael Vazquez for contributing an important step from \eqref{eq:xietax} to \eqref{eq:sys}, the appropriate $2\times2$ representation as a starting point for backstepping design.

The author's problems, ideas, and results were developed with the assistance of Claude and ChatGPT in final theorem formulation, proofs, and drafting throughout the paper, under the author's correction and complete verification.

This work was funded by AFOSR grant FA9550-23-1-0535 and NSF grant ECCS-2151525.

\end{document}